\documentclass[10pt,twocolumn]{IEEEtran}

\IEEEoverridecommandlockouts

\usepackage{amsthm, amssymb}

\def\b{\mathbb}

\def\cc{\mathcal}
\def\sinr{\textrm{SINR}}
\def\snr{\textrm{SNR}}

\def\d{{\rm \ d}}
\def\dint{{\rm \ d}}

\newtheoremstyle{slplain}
  {3pt}
  {3pt}
  {\slshape}
  {}
  {\bfseries}
  {.}%
  { }
  {}

\theoremstyle{slplain}
\newtheorem{thm}{Theorem}
\newtheorem{cor}{Corollary}
\newtheorem{lem}{Lemma}
\newtheorem{pro}{Proposition}

\newtheorem{defi}{Definition}

\usepackage{cite}
\usepackage{amsfonts}
\usepackage{multicol,multienum}
\usepackage{subfigure}
\usepackage{multirow}

\ifCLASSINFOpdf
  \usepackage[pdftex]{graphicx}
  \graphicspath{{../pdf/}{../jpeg/}}
  \DeclareGraphicsExtensions{.pdf,.jpeg,.png}
\else
  \usepackage{float}
  \usepackage[dvips]{graphicx}
  \graphicspath{{../eps/}}
  \DeclareGraphicsExtensions{.eps}
\fi

\usepackage[cmex10]{amsmath}
\usepackage{algorithm}
\usepackage{algorithmic}
\usepackage{array}
\usepackage{url}
\usepackage{setspace}

\begin{document}

\title{Modeling, Analysis and Optimization \\of Multicast Device-to-Device Transmission}

\author{
\IEEEauthorblockA{Xingqin Lin, Rapeepat Ratasuk, Amitava Ghosh and Jeffrey G. Andrews}
\thanks{Xingqin Lin and Jeffrey G. Andrews are with Department of Electrical $\&$ Computer Engineering, The University of Texas at Austin, USA. (Email: xlin@utexas.edu, jandrews@ece.utexas.edu). Rapeepat Ratasuk and Amitava Ghosh are with Nokia Siemens Networks. (Email: {rapeepat.ratasuk, amitava.ghosh}@nsn.com). This research was supported by Nokia Siemens Networks. Date revised: \today. }
}

\maketitle

\begin{abstract}
Multicast device-to-device (D2D) transmission is important for applications like local file transfer in commercial networks and is also a required feature in public safety networks. In this paper  we propose a tractable baseline multicast D2D model, and use it to analyze important multicast metrics like the coverage probability, mean number of covered receivers and throughput. In addition, we examine how the multicast performance would be affected by certain factors like mobility and network assistance. Take the mean number of covered receivers as an example. We find that simple repetitive transmissions help but the gain quickly diminishes as the repetition time increases. Meanwhile, mobility and network assistance (i.e. allowing the network to relay the multicast signals) can help cover more receivers. We also explore how to optimize multicasting, e.g. by choosing the optimal multicast rate and the optimal number of retransmission times.
\end{abstract}



\section{Introduction} 

\subsection{Background}

Recently, there has been a surge of increased interest in supporting direct device to device (D2D)
communication; examples include LTE D2D \cite{3gppD2D} and WiFi Direct \cite{wifiDirect}. Direct D2D connectivity is mainly motivated by the trend of proximity-based services, particularly social networking applications \cite{Corson2012toward}. Also, as the technology of choice for public safety networks in e.g. USA, LTE is supposed to support direct D2D connectivity. From a technical perspective, incorporating D2D in cellular networks opens up many potential benefits for operators. For example, local D2D enables traffic offloading from the core network and communication out of network coverage. Due to the proximity, D2D user equipment (UE) may enjoy very high data rates, low delays and  improved energy efficiency \cite{Doppler2009d2d, lin2013optimal, Fodor2012Design}.


In D2D-enabled cellular networks, direct multicast transmissions, where the same packets from a UE are sent to multiple receivers, are important for scenarios such as the following.

1) \textit{Local file transfer/video streaming}: Local UEs may have common packets for nearby receivers; for example, local marketers may send the same advertising messages to people who happen to be in the neighborhood.

2) \textit{Device discovery}, referring to the process of detecting surrounding devices, is a basic function for many D2D use cases \cite{Fodor2012Design, baccelli2012design}. During device discovery, a discovering device periodically broadcasts beacons to announce its existence, while other devices periodically scan and each may respond to this message once it receives the beacon. 

3) \textit{Cluster head selection/coordination}: For out-of-coverage D2D, it is being discussed in 3GPP to have one UE act as cluster head within a group of UEs: The cluster head can help achieve local synchronization, manage radio resources and schedule transmissions. Cluster head selection normally involves multicast when potential cluster heads send out beacons to announce their roles. 

4) \textit{Group/broadcast communications}: In public safety networks providing services like police, fire and ambulance, D2D group/broadcast communications are required features \cite{3gppD2D}.

In the aforementioned scenarios, compared to communicating with each receiver separately, one direct multicast transmission reduces overhead and saves resources. However, unlike the more commonly studied unicast D2D (see e.g. \cite{yu2011resource, lin2013comprehensive} and references therein), multicast D2D has its own challenges. For example, due to the heterogeneous locations of receivers and complicated radio environment, link quality may vary significantly over receivers in each multicast cluster; thus retransmissions are often required to cover more or all the receivers, which degrades the whole point of multicast vs. unicast. In addition to receiver heterogeneity, it is the UEs rather than base stations (BS) that perform multicast; this introduces additional challenge due to the limited capability of UEs. Despite these challenges, compared to multicast in ad hoc networks, multicast D2D 
has certain conveniences; for example, it may be assisted by the cellular network infrastructure which is not available to ad hoc networks.

It is the significance and distinctive traits of multicast D2D described above that motivate our study in this paper.

\subsection{Related Work}

Multicast in cellular networks can be broadly classified into two classes: Single-rate and multi-rate \cite{richard2012multicast}. In single-rate multicast, the transmitter sends the packets to all the receivers at a common rate \cite{gopala2005throughput, liu2008dynamic, won2009multicast, xu2010adaptive}. For example, in \cite{gopala2005throughput} multicast throughput-delay tradeoff is studied in a single cell system by selecting the median throughput as the multicast rate. In \cite{liu2008dynamic}, dynamic power and subcarrier allocation is performed to adapt to the receiver with the weakest link. In contrast, receiver heterogeneity is exploited in multi-rate multicast, where different receivers in the same multicast cluster may receive packets at different rates based on e.g. the link qualities \cite{suh2005efficient, deb2008real, hou2009cooperative, han2011energy, shao2011layered}. Though being more efficient, multi-rate multicast is much more complex than single-rate multicast in terms of both analysis and implementation. 

In parallel with the academic studies, standardization effort in addressing multicast services has been/is being undertaken and mainly focuses on single-rate multicast. For example, multicast services were addressed in GSM/WCDMA and are being addressed in LTE by 3GPP; the 3GPP work item is known as multimedia broadcast and multicast service (MBMS) \cite{3gppMBMS}. Similarly, 3GPP2 addressed multicast services in CDMA2000 with the work item known as broadcast and multicast service (BCMCS) \cite{3gpp2BCMCS}.

There also exists much work on multicast in ad hoc networks  \cite{chaporkar2005wireless, li2009multicast, shakkottai2010multicast,  liu2011multicast}. For example, in \cite{chaporkar2005wireless} the tradeoff between throughput, stability, and packet loss is studied and a transmission policy is proposed to maximize throughput subject to stability and packet loss constraints.  While \cite{li2009multicast, shakkottai2010multicast} respectively study transport capacity for single hop and multihop wireless networks, \cite{liu2011multicast} tackles ad hoc multicast from the transmission capacity perspective \cite{weber2012TC}.

Unlike the aforementioned studies, there exists a small set of work on multicast in hybrid networks consisting of both ad hoc nodes and cellular infrastructure \cite{park2005enhancing, mao2008multicast}. Though receivers with good channel qualities may relay the multicast traffic to  receivers with poor link qualities using ad hoc mode in \cite{park2005enhancing}, the multicast transmitter is still the BS. In contrast, \cite{mao2008multicast} studies the multicast transport capacity of a hybrid network, and sheds light on its asymptotic growth rate in the number of network nodes. In addition to theoretical analysis, there exist works like \cite{seppala2011network, du2012compressed} which rely more on simulations to understand the performance of multicast D2D.

\subsection{Contributions and Outcomes}

The main contributions and outcomes of this paper are summarized as follows.

\subsubsection{A tractable hybrid network model}

In Section \ref{sec:model}, we introduce a tractable hybrid network model consisting of both ad hoc nodes and cellular infrastructure, which extends our previous unicast D2D model \cite{lin2013comprehensive} to capture the multicast receiver heterogeneity and retransmissions. This model is based on Poisson point processes (PPP), which are highly tractable, and in many practical scenarios, also quite accurate \cite{andrews2011tractable, lee2012stochastic}.

\subsubsection{Multicast performance analysis}

Unlike in one-shot transmission, there exists significant correlation among the signals (resp. interference) over the multicast retransmissions. By tackling this time correlation, we characterize the coverage probability at a particular receiver, and also study how they interact. Building on the coverage analysis, we derive expressions 
for the mean number of covered receivers in each multicast cluster. The expressions allow efficient numerical evaluation; some of them are even in closed-form. Further, we explore multicast throughput and use it as a metric for selecting the optimal multicast rate. These studies reveal the fundamental tradeoff between efficiency  (multicast throughput) and reliability (mean number of covered receivers).

\subsubsection{Impact of mobility}

Though in our default model multicast transmitters are static, we also explore the impact of mobility and analytically show that mobility \textit{hurts} the performance if one would like to support a target SINR for multiple successive transmissions. In contrast, interestingly, we find that mobility \textit{improves} the multicast performance in terms of either coverage probability or mean number of covered receivers or multicast throughput.

\subsubsection{Network-assisted multicast D2D}

We analyze the multicast performance by incorporating network assistance, i.e., allowing the network to relay the multicast signals. It is shown that network assistance can significantly enhance the multicast performance compared to the case of no network assistance. In addition, we formulate a network-assisted multicast D2D optimization problem which minimizes the number of retransmission times subject to a resource constraint at the BSs and a multicast reliability constraint.
An efficient algorithm is also proposed.

\section{System Model}
\label{sec:model}

In this section, we propose a tractable baseline model for studying multicast D2D transmissions.

\subsection{Distributions of Network Nodes}

We consider a hybrid network consisting of both cellular and D2D links. The positions of BSs form an independent Poisson point process (PPP) $\Phi_{\textrm{b}}=\sum_i \delta_{z_i}$ with intensity $\lambda_b$; here $\delta_{z}$ denotes the Dirac measure at position $z \in \mathbb{R}^2$, i.e., for any measurable set $A \subset \mathbb{R}^2$, $\epsilon_{z} (A) = 1$ if $z \in A$, and $0$ otherwise. The PPP model for
BS locations has been recently shown to be about as accurate in terms of both SINR distribution and handover rate as the hexagonal grid for a representative urban cellular network \cite{andrews2011tractable, lin2012towards}.  With a slight abuse of notation, we will also use the position $z$ to indicate the node located at $z$. Similarly, the positions of multicast D2D transmitters form an independent  PPP  $\Phi_{\textrm{m}} = \sum_i \delta_{x_i}$ with intensity $\lambda_m$. These assumptions follow our previous unicast D2D model \cite{lin2013comprehensive}. We further assume that for each D2D transmitter $x_i$, the positions of its intended receivers form a point process $\Phi_{\textrm{m}, x_i} =  \sum_i \delta_{y_i} $ with intensity measure $\Lambda_{x_i} (\cdot) = \lambda_{r} \nu( \cdot \cap B(x_i, R) )$, where $\nu( \cdot )$ is Lebesgue measure in $\mathbb{R}^2$ and $B(x, R)$ denotes the ball centered at $x$ with radius $R$. Note that we do not assume any specific distribution for the receiver point process $\Phi_{\textrm{m}, x_i}$ except the first-order intensity measure; in particular, $\Phi_{\textrm{m}, x_i}$ does not have to be Poisson distributed.

Conditioning on $\Phi_m$, $\{\Phi_{\textrm{m}, x_i}\}$ are assumed to be independent. Those familiar with stochastic geometry will immediately recognize that $\{\Phi_{\textrm{m}, x_i}\}$, which are in the space of point processes on $\mathbb{R}^2$, are independent marks of the PPP $\Phi_{\textrm{m}}$ \cite{baccelli2009stochastic}. Fig. \ref{fig:1} illustrates a snapshot of the spatial distribution of network nodes under the above assumptions. Throughput this paper, the  parameters used in plotting numerical results or simulations are summarized in Table \ref{tab:sys:para} unless otherwise specified.

\begin{figure}
\centering
\includegraphics[width=8cm]{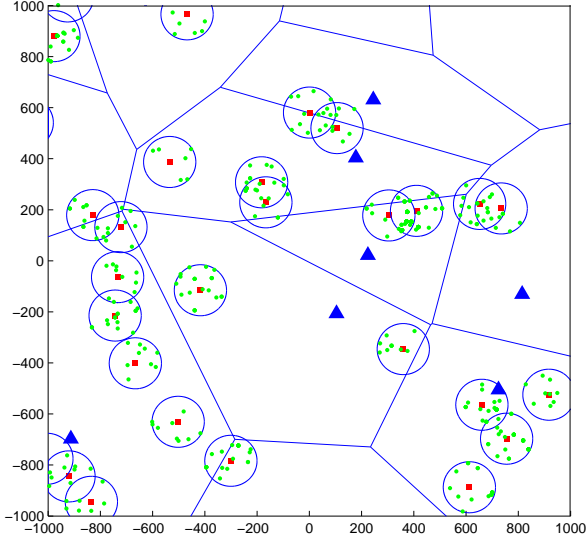}
\caption{A sample realization of the network nodes: Blue solid triangles, red solid squares and green dots denote BSs, multicast D2D transmitters and receivers, respectively. }
\label{fig:1}
\end{figure}

\begin{table}
\centering
\begin{tabular}{|l||r|} \hline
BS Density $\lambda_b$   & $(\pi 500^2)^{-1}$  \\ \hline 
D2D Tx Density $\lambda_m$   & $5 \times (\pi 500^2)^{-1}$  \\ \hline 
D2D Rx Density $\lambda$   & $500 \times (\pi 500^2)^{-1}$  \\ \hline 
Path Loss Exponent $\alpha$ & $3.5$ \\ \hline
Detection Threshold $T$ & $-3$dB \\ \hline  
BS Tx Power $P_c$ & $40$W \\ \hline
D2D Tx Power $P_m$ & $200$mW \\ \hline
Noise PSD & $-174$dBm \\ \hline 
Noise Figure & $9$dB \\ \hline 
Channel Bandwidth & $10$MHz \\ \hline 
\end{tabular}
\caption{Simulation/Numerical Parameters}
\label{tab:sys:para}
\end{table}


\subsection{Multicast Transmission}

Each D2D transmitter $x_i$ has a common message for all the intended receivers in $\Phi_{\textrm{m}, x_i}$; the message can be sent for $\tau_m \in \mathbb{N}$ times, where $\tau_m$ is a pre-configured system parameter. Compared to one shot transmission, sending the message $\tau_m >1$ times enables more intended receivers to successfully decode the message. Further, we assume that multicast transmitters are static during the $\tau_m$ transmissions. As link adaptation is often not possible in multicast transmission \cite{richard2012multicast}, we focus on fixed rate multicast transmission; the rate is often chosen to adapt to receivers of worst or median channel qualities \cite{liu2008dynamic, gopala2005throughput}.

When D2D UEs are in coverage, the ground cellular network can assist D2D communications. Specifically, each in-coverage multicast D2D transmitter has a serving BS; normally the serving BS is the BS providing the strongest reference signal receiving power (RSRP). In the current set-up, this is equivalent to choosing the nearest BS as the serving BS. We use $z_x \in \Phi_{\textrm{b}}$ to indicate the nearest BS of D2D transmitter $x$. Formally,
define the Voronoi cell $\mathcal{C}_{z_i} (\Phi_{\textrm{b}})$ of point $z_i$ with respect to $\Phi_{\textrm{b}}$ as
\begin{align}
\mathcal{C}_{z_i} (\Phi_{\textrm{b}}) = \{ x \in \mathbb{R}^2:  \parallel x - z_i \parallel \ \leq \  \parallel x - z_j \parallel, \forall z_j \in \Phi_{\textrm{b}}  \} .   \notag 
\end{align}
Then each BS $z$ can help D2D transmitters located in its Voronoi cell $\mathcal{C}_{z_i}(\Phi_{\textrm{b}})$ by broadcasting the common messages. Considering the limited time/frequency resource at the BSs, the message of each D2D transmitter $x$ is broadcast by BS $z_x$ at most once.

We assume that D2D is \textit{overlaid} with cellular networks, i.e., D2D transmitters and BSs use orthogonal transmission resources, and thus there is no mutual interference between cellular and D2D transmissions. We refer to \cite{lin2013comprehensive} for \textit{underlay} D2D study, in which D2D and cellular transmissions can be cochannel.
In addition, we assume the multicast message of each D2D transmitter is known by its serving BS. As will become clear from our analysis, the last assumption can be easily relaxed by incorporating the additional hop from the D2D transmitter to its serving BS into the analysis.

\subsection{Channel Model}

Constant transmit powers $P_b$ and $P_m$ are assumed for the BSs and D2D transmitters, respectively. Denote the path loss function as $\ell(r): \mathbb{R}^{+}  \mapsto \mathbb{R}^+$, which is assumed to be continuous and non-decreasing. When concrete results are desired,  we will assume a specific path loss function $\ell(r) = A r^{\alpha}$   where $A>0$ is a constant and  $\alpha >2$ is the path-loss exponent.

Focusing on the signal emitted by the typical transmitter $x_0$ located at the origin, i.e., $x_0 = o$, the received signal $Y_y (n)$ at the receiver $y \in \Phi_{\textrm{m}, o}$ at time $n$ can be written as
\begin{align}
Y_y (n) &= \ell^{-1}(\|y\|) \sqrt{P_m} H_{y,o} (n) X_{o}  + I_y (n) + Z_{y}(n),  \notag 
\end{align}
where $X_{x}$ denotes the signal sent by the D2D transmitter $x$ and $\mathbb{E}[ \|X_{x}\|^2 ] = 1$, $H_{y,x}(n)$ denotes the fading of the link from $x$ to $y$ at time $n$ and is independently distributed as $\mathcal{CN}(0,1)$, $Z_{y}(n)$ denotes the additive noise at receiver $y$ at time $n$ and is independently distributed as $\mathcal{CN}(0,\sigma^2)$, and $I_y (n)$ denotes the aggregate interference  at receiver $y$ at time $n$ and is given by
$
I_y (n) = \sum_{x \neq o} \ell^{-1}(\|y-x\|) \sqrt{P_m} H_{y,x} X_{x}.  
$
Then the signal-to-interference-plus-noise ratio (SINR) of the link from the typical D2D transmitter $x_0=o$ to D2D receiver $y$ at time $n$ equals
$$
\sinr_{y,x_0} (n) = \frac{  F_{y,x_0} (n) / \ell ( \|  y  \| ) }{ \snr^{-1}  +\sum_{j \neq 0} F_{y,x_j} (n) / \ell ( \| x_j - y   \| ) },
$$
where $F_{y,x} = |H_{y,x}|^2 \sim \textrm{Exp}(1)$, and $\snr^{-1} = \sigma^2/P_m$.

Similarly, the received downlink signal $Y_y^{(c)} $ at the receiver $y \in \Phi_{\textrm{m}, o}$ can be written as
\begin{align}
Y_y^{(c)} &= \ell^{-1}(\|z_o\|) \sqrt{P_b}  H_{y,z_o} X_{o}  + I^{(c)}_y + Z_{y},  \notag 
\end{align}
where the aggregate downlink interference $I^{(c)}_y (n)= \sum_{x \neq o} \ell^{-1}(\|y-z_x\|) \sqrt{P_b} H_{y,z_x} X_{x}.  $
The SINR of the link from the nearest BS $z_o$ of the typical D2D transmitter $x_0$ to D2D receiver $y$ equals
$$
\sinr^{(c)}_{y,z_o} = \frac{  F_{y,z_o} / \ell ( \| z_o -  y  \| ) }{ \snr_c^{-1}  +\sum_{x \neq o} F_{y,z_x} / \ell ( \| z_x - y   \| ) },
$$
where $\snr_c^{-1} = \sigma^2/P_b$.

\subsection{Performance Metrics}

From the perspective of analysis, it suffices to consider the typical multicast cluster with $x_0=o$ since, as justified by Palm theory \cite{baccelli2003elements}, its performance indicates the \textit{spatially averaged} performance over all the clusters. Focusing on the typical cluster, we are first interested in the probability that an arbitrary receiver $y \in \Phi_{m,o}$ can decode the multicast message of the typical D2D transmitter $x_0$; we term this coverage probability. Without network assistance, we say the receiver $y \in \Phi_{m,o}$ is covered if $\exists n \in \{1,2,...,\tau_m\}$ such that
$
\sinr_{y,x_0} (n) \geq T,
$
where $T$ is the detection threshold of the fixed rate multicast transmission and is normally greater than $-6$ dB in LTE. 
Formally, denoting $E_n (y)=\{\sinr_{y,x_o} (n) \geq T\}$, the coverage probability at $y$ without network assistance equals
\begin{align}
p(y) \triangleq \mathbb{P}^o \left( \cup_{n =1}^{\tau_m} E_n (y) \right),
\end{align}
where $\mathbb{P}^o(\cdot)$ is  the Palm probability associated with the multicast transmitter process $\Phi_m$. For later use, we define $p_n (y) \triangleq \mathbb{P}^o \left( \cap_{m =1}^{n} E_m (y) \right)$.

Similarly, with network assistance, we say the receiver $y \in \Phi_{m,o}$ is covered if either $\exists n \in \{1,2,...,\tau_m\}$ such that
$
\sinr_{y,x_0} (n) \geq T
$
or $\sinr^{(c)}_{y,z_o} \geq T$. Formally, denoting $E^{(c)} (y)=\sinr^{(c)}_{y,z_o} \geq T$, the coverage probability at $y$ with network assistance equals
\begin{align}
\tilde{p}(y) \triangleq \mathbb{P}^o \left( \cup_{n =1}^{\tau_m} E_n(y) \cup E^{(c)}(y) \right).
\end{align}

While coverage probability characterizes the performance of an individual receiver in the typical cluster, it is also desirable to have a metric to measure the performance of the typical cluster as a whole. Thus, another metric studied in this paper is the mean number of covered receivers in the typical cluster. When network assistance is not available, it equals
\begin{align}
\mathbb{E}^o[ N ] \triangleq  \mathbb{E}^o [ \sum_{y \in \Phi_{m,o}} \mathbb{I} (\{ y \textrm{ is covered} \}) ],
\end{align}
where $\mathbb{I}(E)$ is the indicator function which equals $1$ if the event $E$ is true and 0 otherwise. We use $\mathbb{E}^o[ \tilde{N} ]$ to denote the counterpart of $\mathbb{E}^o[ N ]$ in the case of network assistance.

\section{Multicast without Network Assistance}

In this section we focus on analyzing the multicast performance without network assistance.

\subsection{Coverage Probability}

By definition of Palm probability, the coverage probability at $y$ without network assistance equals
\begin{align}
&p (y  ) = \mathbb{E}^o [ \mathbb{I}(\{ \max_{n=1,...,\tau_m} \sinr_{y,x_0} (n) \geq T   \} ) ] = \notag \\
& \frac{1}{\lambda_m |B|} \mathbb{E} \left[ \int_{x \in B} \mathbb{I}(\{ \max_{n=1,...,\tau_m} \sinr_{y+x,x} (n) \geq T   \} )  \Phi_m(\textrm{d} x) \right], \notag 
\end{align}
where $B$ is an arbitrarily bounded subset of $\mathbb{R}^2$ and $|B|$ denotes its Lebesgue measure. The last relation clearly demonstrates that the coverage performance of the typical cluster indicates the average coverage performance over the clusters. The coverage probability $p (y)$ is explicitly given in Theorem \ref{thm:1}.

\begin{thm}
The probability that the receiver $y \in \Phi_{\textrm{m}, 0}$ is covered by the typical multicast transmitter $x_0 \in \Phi_m$ with $\tau_m$ repetitive transmissions is given by
\begin{align}
p (y  ) 
= \sum_{n=1}^{\tau_m} &(-1)^{n+1} \binom{\tau_m}{n} e^{-n \ell ( \|  y  \| ) T \cdot \snr^{-1} } \cdot \notag \\
&e^{ - 2\pi \lambda_m \int_{0}^\infty ( 1 -   (1 + T\ell ( \|  y  \| )  / \ell ( r ) )^{-n}  ) r \dint r } . 
\label{eq:1}
\end{align}
\label{thm:1}
\end{thm}
\begin{proof}
See Appendix \ref{proof:thm:1}.
\end{proof}

Note that, when $\tau_m$ is large, exact calculation of $p (y)$ based on (\ref{eq:1}) may be cumbersome. Instead, one may consider the following bounds of $p (y)$ which follow from Bonferroni inequalities  (see e.g. \cite{durrett2010probability}):
\begin{align}
p^{(k+1)} (y) \leq p (y  ) \leq p^{(k)} (y), \notag
\end{align}
where $k$ is any odd number in $\{1,...,\tau_m\}$ and $p^{(k)} (y)$ equals the first $k$ summands of the $\tau_m$ summands in (\ref{eq:1}).
By definition, $p(y) = p^{(\tau_m)} (y)$. In general, one gets tighter bounds by making $k$ larger; $p_{k} (y)$ reduces to the union bound when $k=1$. 

Based on Theorem \ref{thm:1}, more specific results can be obtained by plugging explicit path-loss functions $\ell (r)$ in (\ref{eq:1}). For example, for the commonly used path-loss function $\ell (r) = A r^{\alpha}$, the following result immediately follows from Theorem \ref{thm:1}.
\begin{cor}
With $\ell (r) = A r^{\alpha}$, $p (y)$ equals
\begin{align}
\sum_{n=1}^{\tau_m} (-1)^{n+1} \binom{\tau_m}{n} e^{-  n T \cdot \snr^{-1} A \|  y  \|^{\alpha} }  e^{ - \lambda_m K(\alpha, n) T^{\frac{2}{\alpha}} \| y \|^2  }, 
\label{eq:3}
\end{align}
where 
\begin{align}
K(\alpha, n) = \frac{2\pi }{\alpha}\int_0^{\infty} t^{-\frac{2}{\alpha} - 1} \left( 1 - \frac{1}{(1+t)^n} \right) \dint t.
\end{align}
\label{cor:1}
\end{cor}

Fig. \ref{fig:2} shows the coverage probability as a function of detection threshold. As expected, the farther the potential receiver away from the multicast transmitter, the smaller the coverage probability is. Further, repetitive transmissions are instrumental in improving the coverage probability, especially for far away receivers. But the gain diminishes as $\tau_m$ increases.

\begin{figure}
\centering
\includegraphics[width=8cm]{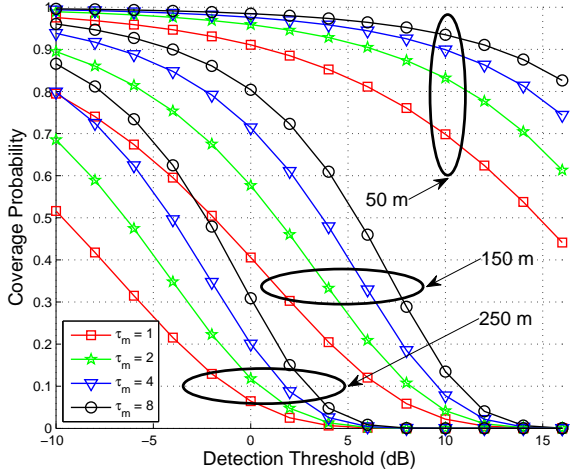}
\caption{Coverage probability versus detection threshold without network assistance: The numbers, $50$ m, $150$ m, $250$ m, indicate three different D2D Tx-Rx distances.}
\label{fig:2}
\end{figure}

Theorem \ref{thm:1} (resp. Corollary \ref{cor:1}) characterizes the coverage probability at a particular receiver, which can be treated as first order coverage performance. As highlighted in \cite{ganti2009spatial, lin2012modeling}, there exist temporal and spatial correlations in the performance at different nodes in a wireless network. Thus, it is of interest to study how the coverage probabilities of different receivers interact, i.e., higher order coverage performance. 
Intuitively,  if some receiver is covered, we may infer that other receivers close to the receiver are also likely to be covered. Towards exploring the correlations of coverage probabilities at different locations, we define the typical covered receiver process
\begin{align}
\tilde{\Phi}_{m,o} = \sum_{y \in \Phi_{m,o}} e_{y} \mathbb{I}(\{ y \textrm{ is covered}\}), \notag 
\end{align}
Obviously, $\tilde{\Phi}_{m,o}$ is a thinning process ``thinned'' from the PPP $\Phi_{m,o}$. However, the thinning operations are \textit{not} independent across the points in $\Phi_{m,o}$ because they are correlated through the multicast transmitter process $\Phi_m$, i.e., due to the presence of common randomness in the locations of the multicast transmitters. This dependent thinning makes the thinning process $\Phi_{m,0}^{(c)}$ no longer a PPP. The following proposition formalizes the correlation concept. 

\begin{pro}
Conditioning on that $y_2$ is covered, the probability that $y_1$ is covered equals
\begin{align}
p (y_1 | y_2 ) &= \sum_{n=1}^{\tau_m} (-1)^{n+1} \binom{\tau_m}{n} p_n (y_1 | y_2   ),
\label{eq:03}
\end{align}
where $p_n (y_1 | y_2)  \triangleq  \mathbb{P}^o(\cap_{m=1}^n E_m(y_1) | \cap_{m=1}^n E_m(y_2) )$ is given by
\begin{align}
&p_n (y_1 | y_2) 
= e^{-n \ell ( \|  y_1  \| ) T \cdot \snr^{-1} } \cdot \notag \\
&e^{ - \lambda_m \int_{\b R^2} \frac{1}{ (1 + \ell ( \|  y_2  \| )  T  / \ell ( \| x - y_2   \| ) )^n } \left(1 -  \frac{1}{ (1 + \ell ( \|  y_1  \| )  T  / \ell ( \| x - y_1   \| ) )^n } \right) \dint x }. \notag
\end{align}
\label{pro:1}
\end{pro}
\begin{proof}
See Appendix \ref{proof:pro:1}.
\end{proof}
Prop. \ref{pro:1} formally shows that the multicast coverage probabilities are indeed correlated across space. Let us examine the summand $p_n (y_1 | y_2)$ in (\ref{eq:03}) to obtain some insight. By definition, it is the conditional probability of the event $\{\min_{k=1,...,n} \sinr_{y_1,x_0} (k)  \geq T\}$ conditional on the event $\{\min_{k=1,...,n} \sinr_{y_2,x_0} (k)  \geq T\}$. Further, direct calculation yields that 
\begin{align}
\frac{p_n (y_1 | y_2)}{p_n (y_1)} = e^{  \lambda_m \int_{0}^{\infty} \prod_{i=1}^2 \left(1 -  \frac{1}{ (1 + \ell ( \|  y_i  \| )  T  / \ell ( \|  x - y_i   \| ) )^n } \right)   \dint x } > 1, \notag
\end{align}
which agrees with intuition: Given the event $\{\min_{k=1,...,n} \sinr_{y_2,x_0} (k)  \geq T\}$, there is a
higher probability that the event $\{\min_{k=1,...,n} \sinr_{y_1,x_0} (k)  \geq T\}$ would happen. The following more specific remarks are in order:
\begin{itemize}
\item The correlation becomes weaker as $\lambda_m$ decreases; in particular, when $\lambda_m$ is asymptotically small, the correlation may be ignored.
\item The correlation becomes weaker as $n$ decreases (which leads to reduced temporal correlation).
\item The correlation becomes stronger when $\| y_1 - y_2  \|$ decreases. In particular, 
$
\lim_{y_2 \to y_1} \frac{p_n (y_1 | y_2)}{p_n (y_1)} = \frac{1}{p_n (y_1)}
$.
\item The correlation becomes stronger as $T$ increases. This is intuitive because with higher $T$ a larger number of interfering nodes come into play. In contrast, when $T$ is small, the outage events at $y_1$ and $y_2$ are respectively dominated by a few nearby interferers around them, and the intersection of the two sets of nearby interferers can be quite small, leading to weak spatial correlation.
\end{itemize}

\subsection{Mean Number of Covered Receivers}

In this subsection we study the mean number of covered receivers in the typical cluster. For concreteness, we focus on the path-loss function  $\ell (r) = A r^{\alpha}$ in the sequel.
\begin{pro}
With $\ell (r) = A r^{\alpha}$, the mean number of covered receivers in the typical cluster is given by
\begin{align}
\mathbb{E}^o [ N ] =& 2\pi  \lambda_r  \sum_{n=1}^{\tau_m}  (-1)^{n+1} \binom{\tau_m}{n} \cdot \notag \\
&\int_0^{R}  r e^{-  n T \cdot \snr^{-1} A r^{\alpha} } \cdot 
e^{ - \lambda_m K(\alpha, n) T^{\frac{2}{\alpha}} r^2  }  \dint r.
\label{eq:2}
\end{align}
\label{pro:2}
\end{pro}
\begin{proof}
See Appendix \ref{proof:pro:2}.
\end{proof}

To gain insight from Prop. \ref{pro:2}, we next focus on a few special cases and/or asymptotic results which have simpler expressions.

\subsubsection{No noise} 
In this case we assume that interference is a dominant issue and thus noise is ignored, i.e., $W\equiv 0$. Then the following corollary follows from Prop. \ref{pro:2}.

\begin{cor}
With  $W\equiv 0$ and $\ell (r) = A r^{\alpha}$, the mean number of covered receivers in the typical multicast cluster  is given by
\begin{align}
\mathbb{E}^o [ N ] = &\frac{\pi \lambda_r}{T^{\frac{2}{\alpha}} \lambda_m}   \sum_{n=1}^{\tau_m}  (-1)^{n+1} \binom{\tau_m}{n} \cdot \notag \\
&  \frac{1}{ K(\alpha, n) }  \left( 1- e^{ - \lambda_m K(\alpha, n) T^{\frac{2}{\alpha}} R^2  } \right).\notag 
\end{align}
In particular, as $\lambda_m \to \infty$, $\mathbb{E}^o [ N ] \sim  \frac{\pi  \tilde{K} (\alpha, \tau_m)\lambda_r}{T^{\frac{2}{\alpha}} \lambda_m}$,
where 
$
\tilde{K} (\alpha, \tau_m) = \sum_{n=1}^{\tau_m} (-1)^{n+1} \binom{\tau_m}{n} \frac{1}{ K(\alpha, n) }
$.
\label{cor:2}
\end{cor}

It follows from Corollary \ref{cor:2} that, when decoding threshold $T$ or cluster size $R$ is small,\footnote{Here we do not consider the case that $\lambda_m$ is small; small $\lambda_m$ makes the assumption that the network is interference-limited invalid.}
\begin{align}
\mathbb{E}^o [ N ]
&\sim \lambda_r \pi R^2 \sum_{n=1}^{\tau_m} (-1)^{n+1} \binom{\tau_m}{n} =\lambda_r \pi R^2, \notag
\end{align}
i.e., all the receivers in the typical cluster can be covered in an expectation sense, agreeing with intuition. Further, $\mathbb{E}^o [ N ]$ is independent of $\lambda_m$ and $\tau_m$. The last fact implies that a single multicast transmission is optimal when $T$ or $R$ is small enough. 

In the extreme case with $\lambda_m \to \infty$, $\mathbb{E}^o [ N ]$ is inversely proportional to the multicast transmitter density $\lambda_m$. Note that the number of repetitions $\tau_m$ does not change the scaling law of $\mathbb{E}^o [ N ]$ (with respect to $\lambda_m$); instead, $\tau_m$ affects $\mathbb{E}^o [ N ]$ only up to the multiplicative factor $\tilde{K} (\alpha, n)$.

\subsubsection{$\alpha = 4$} 
In this case $\mathbb{E}^o [ N ]$ in (\ref{eq:2}) reduces to the following:
\begin{align}
\mathbb{E}^o [ N ]
=& \frac{\pi^{\frac{3}{2}}  \lambda_r}{\sqrt{C_1}} e^{\frac{C_2^2}{4C_1}} \sum_{n=1}^{\tau_m} (-1)^{n+1} \binom{\tau_m}{n}    \notag \\
&\left( Q \left(\frac{C_2}{\sqrt{2 C_1}} \right) - Q \left(\sqrt{2 C_1} R^2 + \frac{C_2}{\sqrt{2 C_1}} \right)  \right), \notag 
\end{align}
where $Q(x) = \frac{1}{\sqrt{2\pi}} \int_{x}^{\infty} e^{-t^2/2} \dint t$, $C_1 = A T \cdot \snr^{-1}$ and $C_2 = \lambda_m K(\alpha, n) T^{\frac{2}{\alpha}}$. This gives a quasi-closed form expression for $\mathbb{E}^o [ N ]$ as $Q(x)$ can be numerically evaluated quite easily.

\subsubsection{$\lambda_m$ is asymptotically small}
In this case, using bounded convergence theorem and binomial theorem, 
\begin{align}
\lim_{\lambda_m \to 0} \mathbb{E}^o [ N ]
=&  \pi \lambda_r  \int_0^{R^2} \left( 1 - (1-  e^{-   T \cdot \snr^{-1} A t^{\frac{\alpha}{2}} }  )^{\tau_m} \right) \dint t .  \notag
\end{align}
As $\tau_m$ increases, the above integrand converges to $1$ at a geometric rate and thus the mean number of covered receivers approaches to $\lambda_r \pi R^2$ very quickly. This fact implies that a very small number of repetition transmissions suffices in sparse networks.

Fig. \ref{fig:3} shows the mean number of covered receivers (normalized by all the potential receivers) as a function of multicast times. Again, repetitive transmissions are instrumental but the gain quickly diminishes as $\tau_m$ increases. This implies that if a D2D transmitter would like to cover far away receivers, other approaches rather than simple repetitive transmissions  are expected; such approaches may include increasing transmit power and interference cancellation. 

\begin{figure}
\centering
\includegraphics[width=8cm]{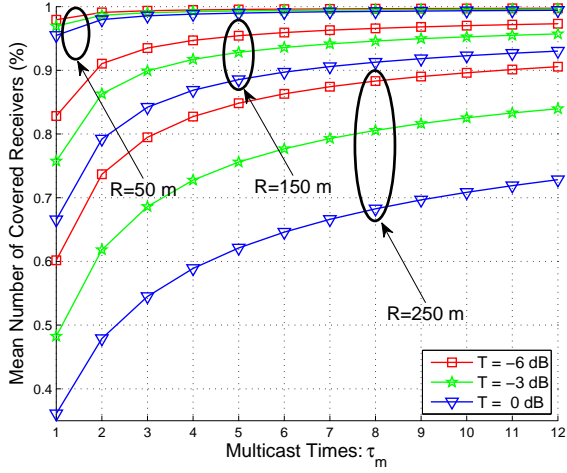}
\caption{Normalized mean number of covered receivers versus multicast times without network assistance.}
\label{fig:3}
\end{figure}

Thus far we have characterized the mean number of covered receivers in the typical cluster. Other properties may be studied with further assumption on the receiver point processes $\{\Phi_{\textrm{m}, x_i}\}$. For concreteness, assume $\{\Phi_{\textrm{m}, x_i}\}$ are Poisson distributed. Then, due to randomness, all the receivers in some clusters may be far away from the multicast transmitter; in an extreme case, there may be no receiver at all in some clusters. We term them \textit{null receiver clusters}. It is of interest to quantify the fraction of null receiver clusters. To this end, we first formalize the concept of null receiver cluster.
\begin{defi}
A multicast cluster is called null receiver cluster if all the receivers have a distance farther than a pre-defined threshold distance $R_{th}$ to the transmitter.
\end{defi} 

A possible criterion for threshold distance $R_{th}$ may be as follows.
$$
\frac{P_m \ell(r) }{N} \geq T , \quad \forall r \leq  R_{th}.
$$
This criterion implies that, without considering interference and fading, a receiver cannot be covered if its distance from the transmitter is farther than $R_{th}$ due to the weak signal. It follows that
$R_{th} = \ell^{(-1)}(\snr^{-1}T)$
where $ \ell^{(-1)}(\cdot)$ denotes the inverse function of $\ell(\cdot)$.

\begin{pro}
The fraction of null receiver clusters equals $e^{ -\lambda_r \pi (\max(R_{th}, R))^2 }$.
\label{pro:8}
\end{pro}
\begin{proof}
See Appendix \ref{proof:pro:8}.
\end{proof}
Note that, conditioning on $\Phi_m$, if $\{\Phi_{\textrm{m}, x_i}\}$ are i.i.d sampled over time, the fraction of null receiver clusters can also be interpreted as the fraction of time that an arbitrary cluster is a null receiver cluster.

\subsection{Multicast Throughput}

Repetition transmission helps improve multicast reliability with increased coverage probability and number of covered receivers. However, repetition consumes more degrees of freedom and thus hurts the throughput. In other words, there exists a fundamental tradeoff between efficiency and reliability. In this section, we explore multicast efficiency. To this end, we define \textit{multicast throughput} (denoted by $\xi$) as follows.

\begin{defi}
Multicast throughput is defined as the mean of the sum rate of all the receivers in the typical multicast cluster. Mathematically,  
\begin{equation}
\xi = \mathbb{E}^o [ N ]  \cdot \frac{1}{\tau_m} \log(1+T). 
\label{eq:xi}
\end{equation}
\end{defi} 

Multicast throughput may serve as a sensible objective for choosing appropriate multicast rate, i.e., $T$. On the one hand, with higher $T$, more sophisticated modulation and coding scheme can be supported and thus higher rate may be achieved. On the other hand, higher $T$ reduces the number of receivers that can be covered by the multicast transmitter. The definition of multicast throughput takes both factors into account by combining $\log(1+T)$ and $\mathbb{E}^o [ N ]$. So we may optimize multicast rate by maximizing the multicast throughput:
\begin{align}
\textrm{maximize}_{T > 0} \quad \mathbb{E}^o [ N ]  \cdot \frac{1}{\tau_m} \log(1+T), 
\label{eq:rateOpt}
\end{align}
where $\mathbb{E}^o [ N ]$ is explicitly given in Prop. \ref{pro:2}. The above optimization is of single variable and thus can be solved efficiently. More explicit results may be obtained under special cases; for example, the following proposition considers the case with noise ignored and $\lambda_m \to \infty$. 

\begin{pro}
With $W \equiv 0$ and $\lambda_m \to \infty$, multicast rate optimization (\ref{eq:rateOpt}) reduces to
\begin{align}
\textrm{maximize}_{T > 0} \quad T^{-\frac{2}{\alpha}} \log(1+T).
\end{align}
Further, it has a unique optimal point $T^\star > \frac{\alpha}{2} -1$ that equals the unique solution of the equation: 
$
\frac{x}{1+x} = \frac{2}{\alpha} \log(1+x).
$
\label{pro:9}
\end{pro}
\begin{proof}
See Appendix \ref{proof:pro:9}.
\end{proof}

To gain some intuition, we show multicast throughput as a function of $T$ in Fig. \ref{fig:7}. It is shown that the optimal rate $T^\star$ is relatively robust to $\tau_m$; for example, with $\alpha = 3.5$, optimal $T^\star$ is around $7$ dB for either  $\tau_m = 1$ or $\tau_m = 4$. It is also shown that higher multicast throughput is obtained with median pathloss exponent, agreeing with intuition: High pathloss exponent provides better spatial separation in terms of interference but also leads to high loss of signal power; whereas the converse is true with low pathloss exponent.

\begin{figure}
\centering
\includegraphics[width=8cm]{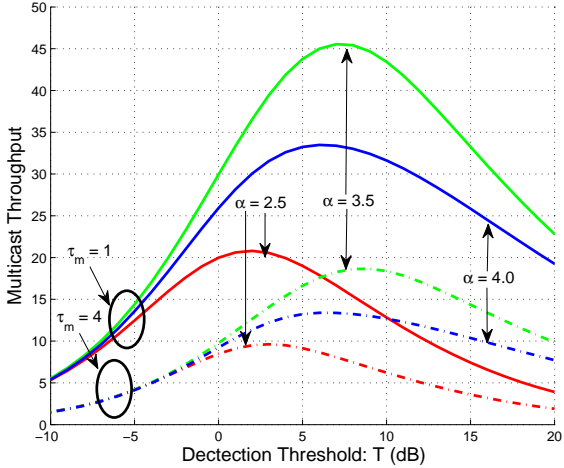}
\caption{Multicast throughput versus detection threshold: $R = 150$ m.}
\label{fig:7}
\end{figure}

In (\ref{eq:xi}), as $\tau_m$  increases, $\mathbb{E}^o [ N ]$ increases but $\frac{1}{\tau_m} \log(1+T)$ decreases. As the latter typically dominates the former, the defined multicast throughput $\xi$ decreases with $\tau_m$. We illustrate the tradeoff between efficiency -- multicast throughput $\xi$ -- and reliability -- mean number of covered receivers $\mathbb{E}^o [ N ]$ in Fig. \ref{fig:8}. How to strike a balance between efficiency and reliability depends on the application scenarios. Nevertheless, the bottom line may be that reliability should not be stressed to an extent such that multicast loses its superiority over unicast. For example, ignoring overhead issues, reasonable choice of $\tau_m$ should satisfy the follow relation:
$$
\xi  \geq \mathbb{E}^o \left[  \frac{\tau_m}{|\Phi_{m,x_0}|}  \sum_{y \in \Phi_{m,x_0}}  p(y) \log(1+T)   \right],
$$
where the right hand side denotes the achievable sum rate if the typical transmitter unicasts to each receiver separately.

\begin{figure}
\centering
\includegraphics[width=8cm]{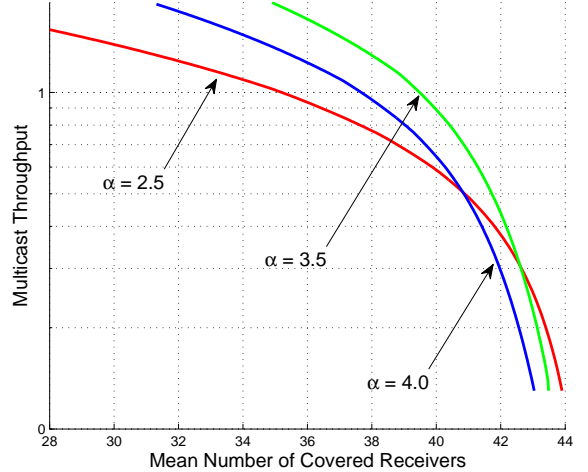}
\caption{Tradeoff between efficiency -- multicast throughput -- and reliability -- mean number of covered receivers.}
\label{fig:8}
\end{figure}

\subsection{Impact of Mobility}

Recall that in our default model multicast transmitters are static. Correspondingly, in the previous analysis on coverage probability and mean number of covered receivers, we first perform time average by fixing the spatial realization of $\Phi_m$; then we de-condition on $\Phi_m$ to average out the spatial randomness. A natural question arises: What is the impact of mobility? To answer this question, we assume in this section that the multicast transmitter process $\Phi_m$ is independently re-sampled at each time slot during the multicast transmissions, i.e., $\{\Phi_m(n)\}$ are independent PPPs. This assumption can model the scenario where the multicast transmitters are of high mobility. The static and high mobility cases represent two extreme mobility pattern; the real mobility pattern lies somewhere in between \cite{lin2012towards}.

Surprisingly, based on the results for static scenario, the characterization of the performance of high mobility scenario is quite clean, as given in the following Prop. \ref{pro:3}.

\begin{pro}
With high mobility assumption and path-loss function $\ell (r) = A r^{\alpha}$, the coverage probability $p (y  )$ and mean number of covered receivers $\mathbb{E}^o [ N ]$ are respectively given by (\ref{eq:3}) and (\ref{eq:2}) but with $K(\alpha, n)$ replaced by $n K(\alpha, 1)$. 
\label{pro:3}
\end{pro}
\begin{proof}
See Appendix \ref{proof:pro:3}.
\end{proof}

To get some insight about how mobility affects multicast efficiency, let us recall in the static case
$
\log 1/p_n (y  ) 
$
is proportional to $K(\alpha, n)$; in the high mobility case $
\log 1/ p_n (y  ) 
$
is proportional to $n K(\alpha, 1)$.  The following Lemma \ref{lem:1} shows that $n K(\alpha, 1)$ is greater than $K(\alpha, n)$ except the trivial case $n=1$. It follows that $p_n (y), n >1$, in the static case is larger than its counterpart in the high mobility case. In other words, mobility \textit{hurts} the performance if one would like to support a target SINR for $n$ successive transmissions, agreeing with intuition: Motility brings extra randomness to the received SINR and thus makes it harder to successively meet the target SINR. 
\begin{lem}
For any integer $n >1$,
$
n K(\alpha, 1) - K(\alpha, n) > 0.
$
\label{lem:1}
\end{lem}
\begin{proof}
See Appendix \ref{proof:lem:1}.
\end{proof}

The impact of mobility on $p (y  )$ or $\mathbb{E}^o [ N ]$ is more subtle. 
Fig. \ref{fig:4} compares the mean number of covered receivers in static network (i.e., our default model) to that of mobile network. Interestingly, it shows that mobility can increase the mean number of covered receivers. Further, the loss due to the static environment (at least at the time scale of $\tau_m$) can be hardly overcome by increasing the number of retransmissions. This is because the signal and interference powers largely depend on the node locations; multiple transmissions may exploit the fast fading but cannot fundamentally change the signal and interference powers. 

\begin{figure}
\centering
\includegraphics[width=8cm]{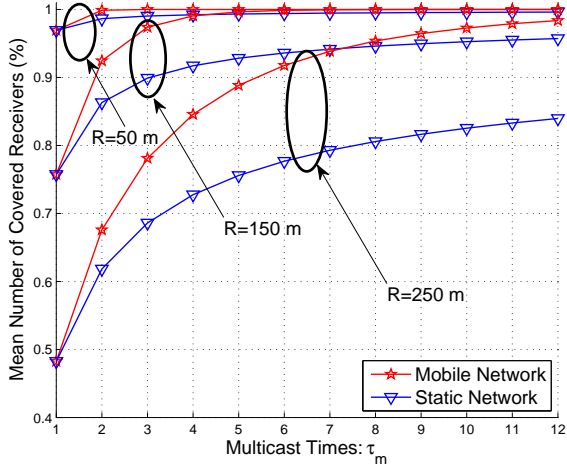}
\caption{Mobility increases the mean number of covered receivers.}
\label{fig:4}
\end{figure}

\section{Multicast with Network Assistance}

In this section we analyze the multicast performance by incorporating network assistance, i.e., allowing the network to relay the multicast signals. Recall that $z_{x_0} = z_o$ denotes the BS that is closest to the typical multicast transmitter. We first study the probability that the receiver located at $y \in \Phi_{\textrm{m}, x_0}$ is covered by the BS $z_o$ in the following Lemma \ref{lem:4}.
\begin{lem}
The probability that the receiver located at $y \in \Phi_{\textrm{m}, x_0}$ is covered by the BS $z_o$ is given by
\begin{align}
p_{c} (y ) 
&= \int_{\mathbb{R}^2} p_{c} (y |  x ) \cdot  \lambda_b e^{ - \lambda_b \pi \| x \|^2 } \d x,
\end{align}
where $p_{c} (y |  x )$ equals
\begin{align}
e^{ - T \ell ( \| x - y  \| ) \snr_c^{-1}  - \int_{ B^{c}(0, \| x \|) }  ( 1 - \frac{1}{   1 + T \ell ( \|   x - y \| ) /   \ell ( \| z - y   \| )  } ) \lambda_b \dint z }.
\end{align}
\label{lem:4}
\end{lem}

The proof of Lemma \ref{lem:4} follows from \cite{andrews2011tractable} and is omitted for brevity. It is noticed that the domains of integrations in Lemma \ref{lem:4} are hard to manipulate  to get more explicit results. To overcome this inconvenience, we shall adopt the following approximation: 
\begin{align}
\|z-y\| \approx \| z - x_0 \|, \forall z \in \Phi_b. 
\label{eq:approx}
\end{align}
The above approximation may be justified when the multicast regions are small compared to the coverage area of each BS.  With the above approximation, the following Corollary \ref{cor:3} can be obtained.
\begin{cor}
With the approximation (\ref{eq:approx}) and $\ell (r) = A r^{\alpha}$, $p_{c} (y) \approx p_{c}, \forall y \in  \Phi_{\textrm{m}, x_0}$, where
\begin{align}
p_{c} 
&\triangleq 2\pi \lambda_b  \int_{0}^{\infty} e^{ - T \cdot \snr_c^{-1} A r^{\alpha} } e^{ - 2\pi \lambda_b H(T,\alpha) r^2 }  e^{ - \lambda_b \pi r^2 } r \dint r ,
\end{align}
and
$
H(T, \alpha) = \int_{1}^{\infty} \frac{x}{1+x^{\alpha}/T}  \dint x
$. In particular,  when $\alpha = 4$, $p_c = \frac{\pi^{\frac{3}{2}}\lambda_b}{2\sqrt{C_3}} e^{\frac{C_4^2}{4C_3}} Q\left( \frac{C_4}{\sqrt{2C_3}} \right)$, where $C_3 = A T \cdot \snr_c^{-1}$ and $C_4 = 2\pi \lambda_b H(T,\alpha) + \pi \lambda_b$; when there is no noise, $p_c = \frac{1}{1+2H(T, \alpha) }$.
\label{cor:3}
\end{cor}

For simplicity we will use equality instead of an approximation in the sequel.
With network assistance, the probability that the receiver $y \in \Phi_{\textrm{m}, x_0}$ is covered as long as either the BS $z_o$ or the multicast transmitter $x_0$ covers it. Further, these two events are independent. It follows that the coverage probability at $y \in \Phi_{\textrm{m}, x_0}$ with network assistance equals
\begin{align}
\tilde{p}_{0} (y) &= 1 - (1 - p_c ) (1 - p_{0} (y)). \notag 
\end{align}
Rearranging the above equality yields the following result.
\begin{pro}
With network assistance, the coverage probability of the receiver $y \in \Phi_{\textrm{m}, 0}$  equals
\begin{align}
\tilde{p} (y) 
& =   p(y) + p_c   (1 -  p (y) ),
\end{align}
where $p (y)$ and $p_c$ are given in Theorem \ref{thm:1} and Corollary \ref{cor:3}, respectively. Accordingly, the mean number of covered receivers equals
\begin{align}
\mathbb{E}^o [ \tilde{N} ] = \mathbb{E}^o [ N ] + p_c ( \lambda_r \pi R^2 - \mathbb{E}^o [ N ]  ).
\end{align}
\label{pro:5}
\end{pro}

Prop. \ref{pro:5} shows that the network assistance is most useful when $\lambda_r \pi R^2 - \b E^o [ N ] \geq 0$ is large. In particular, with moderate to large detection threshold $T$ and cluster range $R$, network assistance can significantly reduce the number $\tau_m$ of transmissions to achieve the same mean number of covered receivers in the absence of network assistance. Fig. \ref{fig:5} shows  the mean number of covered receivers with network assistance as a function of multicast times. As expected, network assistance is very useful; the gain is particularly pronounced in the case of large  multicast radius. In addition, Fig. \ref{fig:5}  shows that the analytical results match the empirical results fairly well; in particular, the approximation (\ref{eq:approx}) used in the case of network assistance analysis does not lead to noticeable accuracy lost, at least from the perspective of mean number of covered receivers.

\begin{figure}
\centering
\includegraphics[width=8cm]{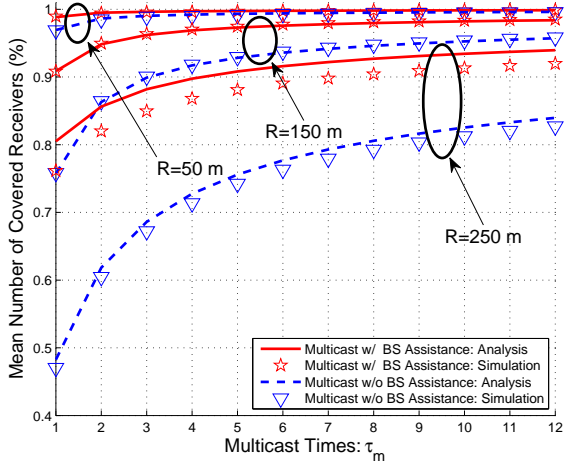}
\caption{Network assistance helps increase the mean number of covered receivers.}
\label{fig:5}
\end{figure}

Further, conditioned on $\| z_o \| = r$,  $\mathbb{E}^o [ \tilde{N} | \| z_0 \| = r ]$ equals
\begin{align}
\mathbb{E}^o [ N ] + e^{ - T \cdot \snr_c^{-1} A r^{\alpha}  - 2\pi \lambda_b H(T,\alpha) r^2 } ( \lambda_r \pi R^2 - \mathbb{E}^o [ N ]  ), \notag 
\end{align}
from which it is clear that the network assistance is most useful when the distance from the multicast transmitter to its nearest BS is not large. How to optimize this network assistance is the subject of the next section.

\section{Optimizing Multicast Transmissions}

In this section we aim to optimize the network assisted multicast transmissions. The overall objective is to seek for optimum network assistance rule to minimize retransmission times while certain network constraints can be satisfied.

For each BS $z$, let $g_{z}: \b R^{+} \to \{0,1\}$ be a mapping such that $g_z( \| x - z \| )=1$ if BS $z$ helps D2D transmitter $x$ located in its cell, i.e.,  $x \in \mathcal{C}_z (\Phi_b)$. As the transmission resources of BSs are limited, we assume each BS $z$ can help at most $B$ multicast sessions in its cell. Mathematically, for $\forall z \in \Phi_b$,
\begin{align}
\sum_{x \in \Phi_m \cap  \mathcal{C}_z (\Phi_b)} g_z ( \| x-z \| ) \leq B,.
\label{eq:10}
\end{align}
From the spatial average perspective, the following constraint is imposed at the typical cell.
\begin{align}
\mathbb{E}^o \left[ \sum_{x \in \Phi_m \cap  \mathcal{C}_o (\Phi_b)} g_o ( \| x \| ) \right] \leq B.
\label{eq:20}
\end{align}
In this section the Palm probability is defined with respect to the BS point process $\Phi_b$ instead of D2D transmitter point process $\Phi_m$; the two Palm distributions may be connected with Neveu exchange formula \cite{baccelli2009stochastic}. By definition, $g_o (\cdot) \in \{0,1\}$. However, under the Palm measure, the performance seen by the typical BS is a spatial average; thus with a slight abuse of notation, we allow $g_o (\cdot) \in [0,1]$. In the sequel, we shall refer to (\ref{eq:10}) (resp. (\ref{eq:20})) as \textit{resource} constraint.

Further, we require that a certain fraction $\eta$ of the intended receivers associated with the D2D transmitters in each cell should be covered. Mathematically, using Corollary \ref{cor:3} and Prop. \ref{pro:5}, we have the following constraint: for $\forall z \in \Phi_b$,
\begin{align}
\frac{\sum_{x \in \Phi_m \cap \mathcal{C}_z (\Phi_b)}  \mathbb{E} [ \sum_{y \in \Phi_{m,x}}  \mathbb{I}( \{y \textrm{ is covered}\}) ] }{ \sum_{x \in \Phi_m \cap \mathcal{C}_z (\Phi_b)}  \mathbb{E}  \left[ | \Phi_{m,x}  |\right] } \geq  \eta,
\label{eq:11}
\end{align}
where the numerator implicitly depends on $g_z (\cdot)$. From the spatial average perspective, the following constraint is required for the typical cell with the BS located at the origin. 
\begin{align}
\frac{ \mathbb{E}^o \left[  \sum_{x\in \Phi_m \cap \mathcal{C}_o (\Phi_b)}  \sum_{y \in \Phi_{m,x}}  \mathbb{I}( \{y \textrm{ is covered}\})  \right] }{  \mathbb{E}^o \left [\sum_{x \in \Phi_m \cap \mathcal{C}_z (\Phi_b)} | \Phi_{m,x}  | \right] }  \geq  \eta.
\label{eq:21}
\end{align}
In the sequel, we shall refer to (\ref{eq:11}) (resp. (\ref{eq:21})) as \textit{reliability} constraint. The following Prop. \ref{pro:6} gives more explicit expressions for the expectation terms involved in (\ref{eq:20}) and (\ref{eq:21}).

\begin{pro}
The three expectation terms in (\ref{eq:20}) and (\ref{eq:21}) are respectively given as follows.
\begin{align}
&\mathbb{E}^o  [ \sum_{x \in \Phi_m \cap   \mathcal{C}_0 (\Phi_b)} g_o ( \| x \| ) ]  = \frac{\lambda_m}{\lambda_b} \mathbb{E}_D \left[ g_o (D)  \right] \notag \\
&\mathbb{E}^o   [\sum_{x \in \Phi_m \cap \mathcal{C}_z (\Phi_b)}  | \Phi_{m,x}  | ] = \frac{\lambda_m}{\lambda_b} \bar{N}_{\max} \notag \\
&\mathbb{E}^o [  \sum_{x\in \Phi_m \cap \mathcal{C}_o (\Phi_b)}  \sum_{y \in \Phi_{m,x}}  \mathbb{I}( \{y \textrm{ is covered}\})  ]   \notag \\
&\quad  \quad \quad \quad \quad =  \frac{\lambda_m}{\lambda_b}  \mathbb{E}_D \left[  h( D ; \tau_m, g_o (D) ) \right] , \notag 
\end{align}
where $D$ is a Rayleigh distributed random variable with pdf
$
f_D (r) = 2\pi \lambda_b r e^{ -\lambda_b \pi r^2 }, r \geq 0
$; $\bar{N}_{\max} = \lambda_r \pi R^2$; and $h: \mathbb{R}^+ \mapsto \mathbb{R}^+$ is given by
$
h(r ; \tau_m, g_o (r) ) = \bar{N} (\tau_m) + g_o ( r ) \cdot q ( r ) ( \bar{N}_{\max} -  \bar{N} (\tau_m)), 
$
where $\bar{N} (\tau_m) =  \mathbb{E}^o[N (\tau_m) ]$ is given in Prop. \ref{pro:2}, $q: \mathbb{R}^+ \mapsto [0,1]$ is defined as
$
q(r) = e^{ - T \snr_c^{-1} A r^{\alpha}  - 2\pi \lambda_b H(T,\alpha) r^2 }.
$
\label{pro:6}
\end{pro}
\begin{proof}
See Appendix \ref{proof:pro:6}.
\end{proof}

Using Prop. \ref{pro:6}, we can cast the spatial averaged multicast optimization problem as follows.
\begin{align}
\textrm{minimize} \ \ & \tau_m  \label{eq:24}  \\
\textrm{subject to} \ \  & \mathbb{E}_D \left[ g_o (D)  \right] \leq \frac{\lambda_b}{\lambda_m} B  \notag \\ 
& \mathbb{E}_D \left[  h( D ; \tau_m, g_o (D) ) \right] \geq \eta \bar{N}_{\max}  \notag \\  
                         & 0\leq g_o (r) \leq 1, \forall r \geq 0.  \notag
\end{align}
This is a mixed integer nonlinear programming which is in general notoriously hard to solve. Worse still, the design space $g_o (\cdot)$ is of infinite dimension; it is not \textit{a priori} clear at all what kind of mapping $g_o (\cdot)$ we ought to pursue. Furthermore, as $g_o (\cdot)$ represents the optimum network assistance rule averaged across the space, it does not lead to readily implementable solution for the network. For these reasons, we are more interested in the following ``online'' problem: Given a realization of $\Phi_b$ and $\Phi_m$, how should each BS $z$ help the D2D transmitters in its cell while satisfying its resource and reliability constraints? Mathematically, each BS $z$ aims to solve the following problem.
\begin{align}
\textrm{minimize } & \ \  \tau_m  \label{eq:25}  \\
\textrm{subject to}   & \sum_{x \in \Phi_m \cap \mathcal{C}_z (\Phi_b)}  g_z ( \| x - z \| ) \leq B  \notag \\  
& \frac{\sum_{x \in \Phi_m \cap \mathcal{C}_z (\Phi_b)}  h\left(\| x - z \| ; \tau_m, g_z (\|x-z\|) \right) }{ | \Phi_m \cap  \mathcal{C}_z (\Phi_b) | \cdot \bar{N}_{\max} }  \geq   \eta   \notag \\ 
                         & g_z ( \| x - z \|  ) \in \{0, 1\}, \forall x \in \Phi_m \cap  \mathcal{C}_z (\Phi_b).  \notag
\end{align}

Though the above problem is still an integer programming, the design space $g_z (\cdot)$ is of finite dimension. In particular, we only need to determine finite number of binary variables, $g_z ( \| x_i - z \|  ), i=1,...,M_z$, where $M_z = | \Phi_m \cap  \mathcal{C}_z (\Phi_b) |$. However,  with an exhaustive search the complexity is still exponential in $M_z$. We next analyze the optimality structure of the problem to design an efficient algorithm. To this end, we first note that there always exists a feasible solution; for example, the solution with $g_z ( \| x - z \|  ) \equiv 0$ but large enough $\tau_m$ is a feasible one.

For ease of exposition, relabel the D2D transmitters (located in the cell of BS $z$) in the order of increasing distance to BS $z$, i.e., $ r_1 \leq ... \leq  r_{M_z} $ where $r_i = \| x_i - z \|$, and let $\tau_m^{\star}$ denote the minimum value that can be obtained in the above problem. Then the following result holds.
\begin{pro}
There exists an optimal solution such that $g^{\star}_z ( \| x_1 - z \|  ) \geq .... \geq  g^{\star}_z ( \| x_{M_z} - z \|  )$.
\label{pro:7}
\end{pro}
\begin{proof}
See Appendix \ref{proof:pro:7}.
\end{proof}

For each possible $\tau_m$, Prop. \ref{pro:7} implies that BS $z$ can focus on the $\min(M_z, B)$ nearest D2D transmitters and assists as few of them as possible to save resources. Further, as the mapping $h(\cdot)$ is moronically increasing with $\tau_m$, we then can use a binary search for the minimum $\tau_m^\star$ over $\{1,2,...,\tau_{\max}\}$, where $\tau_{\max}$ is a large enough integer such that $\frac{1}{M_z} \sum_{i=1}^{M_z} h\left( r_i ;
 \tau_{\max}, 0 \right) \geq \eta \bar{N}_{\max} $. 
The running time of this algorithm  is $O (  M_z \log \tau_{\max} )$; thus for given $M_z$ and $\tau_{\max}$,  the proposed algorithm is efficient. However, we need to find a valid but $\textit{a priori}$ unknown  $\tau_{\max}$ for initialization purpose. With reasonable $\eta$, $\tau_{\max}$'s are usually not large and we can find one quite efficiently. We simulate the proposed algorithm and present the network assistance statistics in Fig. \ref{fig:6}. 
As expected, D2D transmitters that are closer to their nearest BSs have higher chance to get network assistance. 

\begin{figure}
\centering
\includegraphics[width=8cm]{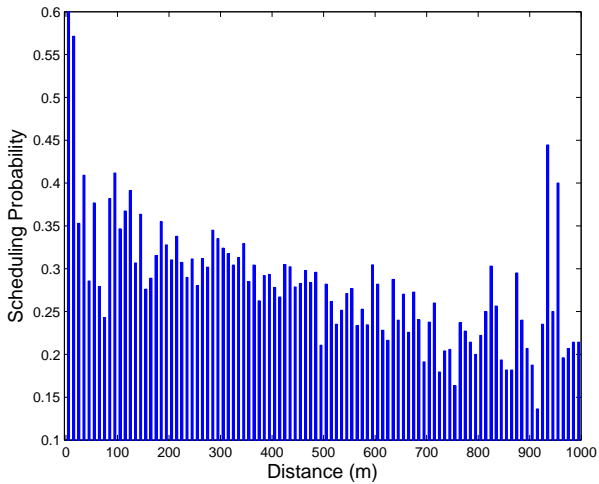}
\caption{Network assistance statistics in the case of optimized multicast transmissions: $\eta = 95\%, B = 2$. The abscissa denotes the distance between D2D transmitter and its nearest BS; given the distance, the associated ordinate value gives the corresponding probability that the D2D transmitter is scheduled by the BS for assistance.}
\label{fig:6}
\end{figure}


Finally, we comment how to construct a reasonably good solution to  the original prohibitively difficult network-wise optimization problem (\ref{eq:24}) by collecting and appropriately averaging the network assistance statistics $\{g_z (r)\}$ across the space as follows. We first simulate a large enough network with area e.g. $B(0, R_n)$ where each BS assists the D2D transmitters in its cell using the proposed algorithm. Then, partition $\b R^+$ into $I$ non-overlapping intervals $[r_i,r_{i+1}), i=0,1,...,I-1$ with $r_0=0, r_I = \infty$, and $|r_{i+1} - r_i | = \Delta$ for $i=0,1,...,I-2$; and collect the statistics as follows: for $\forall r\in [r_i, r_{i+1})$, $\bar{g}_0 (r)$ equals 
\begin{align}
&\frac{\sum_{z \in B(0, R_n)}  \sum_{ x \in \Phi_{m} \cap \mathcal{C}_z (\Phi_b) } 1_{ \| x - z \| \in [r_i,r_{i+1}) } \cdot g_z ( \| x - z \| )}{ \sum_{z \in B(0, R_n)}  \sum_{ x \in \Phi_{m} \cap \mathcal{C}_z (\Phi_b) } 1_{ \| x - z \| \in [r_i,r_{i+1}) } }. \notag 
\end{align} 
In this way, we arrive at a piece-wise constant solution $\bar{g}_0 (\cdot) $ to the problem (\ref{eq:24}). With  $\bar{g}_0 (\cdot) $, we can use binary search for the corresponding minimum objective value $\bar{\tau}_m$, which is expected to be approximately equal to the spatial average obtained from simulation, i.e.,
 \begin{align}
\bar{\tau}_m   \approx \frac{1}{ | \Phi_b ( B(0, R_n) )   | } \sum_{z \in B(0, R_n)} \tau_m^{\star} ( g_z   ) . \notag 
\end{align} 
The above approach can be viewed as simulated annealing which leverages the ergodicity of the underlying random processes.

\section{Conclusions}

In this paper we propose a tractable analytical model for the analysis and design of multicast D2D. The model has been used to analyze important multicast metrics like the coverage probability, mean number of covered receivers and throughput. In addition, how the multicast performance would be affected by certain factors like mobility and infrastructure assistance has also been examined. We have also explored how to optimize multicasting, e.g. by choosing the optimal multicast rate and the optimal number of retransmission times.

This work can be  extended in a number of ways. One may consider more sophisticated network-assisted D2D; for example, BSs may schedule multicast D2D UEs to achieve better interference coordination among D2D transmitters. It is also interesting to extend this work to heterogeneous networks consisting of different types of lower power nodes besides macro BSs, possibly with multiple antenna techniques employed. Last but not least, millimeter wave (mmWave) seems quite promising for D2D if there is Tx-Rx proximity but so far the performance of mmWave D2D is largely open.

\section*{Acknowledgment}

The authors gratefully acknowledge Professor Fran\c{c}ois Baccelli for his help in identifying the problem and detailed technical discussions.

\appendix

\subsection{Proof of Theorem \ref{thm:1}}
\label{proof:thm:1}

The proof consists of two steps: We first perform time average over fading by fixing the spatial realization of $\Phi_m$; then we de-condition on $\Phi_m$ to average out the spatial randomness. This two-step argument used to deal with the temporal correlation of multicast process is motivated by \cite{haenggi2012diversity} which deals with spatial correlation over multiple receive antennas. 

Let $E_I (y) = \bigcap_{n \in I} E_n,  I \subset \{1,...,\tau_m\}$, where we drop the argument of $E_n(y)$ for notional simplicity. Then conditioned on $\Phi_m$, the probability that $y \in B(0, R)$ is covered by the typical multicast transmitter is given by
\begin{align}
p (y | \Phi_m  ) &= \mathbb{P}^o \left( \bigcup_{n=1,...,\tau_m} E_n | \Phi_m  \right) \notag \\
 &=  \sum_{n=1}^{\tau_m} (-1)^{n+1} \sum_{I \subset \{1,...,\tau_m\}: |I| = n} \mathbb{P}^o \left( E_I | \Phi_m  \right), \notag 
\end{align}
where the second equality follows from inclusion-exclusion principle.
Note that conditioned on $\Phi_m$, the events $E_n$, $n=1,...,\tau_m$,  are independent, because the fading fields are assumed to be independent across both space and time. It follows that $\mathbb{P}^o \left( E_I | \Phi_m  \right)$ only depends on the cardinality of $I$, i.e., $\mathbb{P}^o \left( E_I | \Phi_m  \right) \equiv \mathbb{P}^o \left( E_{\{1,...,n\}} | \Phi_m  \right), \forall I \subset \{1,...,\tau_m\}$ with  $|I| = n$. Thus, denoting $p_n (y | \Phi_m  ) = \mathbb{P}^o \left( E_{\{1,...,n\}} | \Phi_m  \right)$, $p (y | \Phi_m  )$ can be further written as
\begin{align}
p (y | \Phi_m  ) &=\sum_{n=1}^{\tau_m} (-1)^{n+1} \binom{\tau_m}{n} p_n (y | \Phi_m  ). \notag
\end{align}
Next we focus on computing $p_n (y | \Phi_m  )$. Due to the independence of the fading fields across time, 
\begin{align}
p_n (y | \Phi_m  ) &= \mathbb{P}^o (  \sinr_{y,0} (k) \geq T, \forall k \in \{1,...,n\}  | \Phi_m  ) \notag \\
&= \prod_{k=1}^{n} \mathbb{P}^o (  \sinr_{y,0} (k) \geq T   | \Phi_m  ), \notag 
\end{align}
where $\mathbb{P}^o (  \sinr_{y,0} (k) \geq T   | \Phi_m  )$ equals
\begin{align}
&\mathbb{P}^o \left(  \frac{  F_{y,x_0} (k) / \ell ( \|  y  \| ) }{ \snr^{-1}  +\sum_{j \neq i} F_{y,x_j} (k) / \ell ( \| x_j -  y   \| ) } \geq T   \bigg | \Phi_m  \right ) \notag \\
&= \mathbb{P}^o (   F_{y,x_0} (k)   \geq \ell ( \|  y  \| ) T (\snr^{-1}   \notag \\
& \quad +   \sum_{j \neq 0} F_{y,x_j} (k) / \ell ( \| x_j - y   \| ) )  | \Phi_m  ) \notag \\
&= e^{- \ell ( \|  y  \| ) T \cdot \snr^{-1} } \mathbb{E}^o_{F} [ e^{ -  \ell ( \|  y  \| )  T \sum_{j \neq 0} F_{y,x_j} (k) / \ell ( \| x_j - y   \| ) }   | \Phi_m ] \notag  
\end{align}
where the last equality is due to $F_{y,x_0} (k) \sim \textrm{Exp}(1)$. Further, by Slyvnyak's theorem \cite{baccelli2009stochastic}, independence of the fading fields across space and the Laplace transform of $F\sim \textrm{Exp}(1)$ (which is $\cc L_{F} (s) = \frac{1}{1+s}$), $p_n (y | \Phi_m  ) $ equals $e^{- n\ell ( \|  y  \| ) T \cdot \snr^{-1} }$ multiplied by
\begin{align}
&\prod_{k=1}^{n} \mathbb{E}_{F} \left[ e^{ - n \ell ( \|  y  \| )  T \sum_{j} F_{y,x_j} (k) / \ell ( \| x_j - y   \| ) }   | \Phi_m  \right ] \notag \\
&=  \prod_j \prod_{k=1}^{n}  \mathbb{E}_{F} \left[ e^{ -  \ell ( \|  y  \| ) T  F_{y,x_j} (k) / \ell ( \| x_j - y   \| ) }   | \Phi_m  \right ] \notag \\
&=  \prod_j    \frac{1}{ (1 + \ell ( \|  y  \| )  T  / \ell ( \| x_j - y   \| ) )^n }. \notag 
\end{align}
Now de-conditioning with respect to $\Phi_m$ yields
\begin{align}
&p_n (y   ) = \mathbb{E}_{\Phi_m} [p_0^{(n)} (y | \Phi_m  )  ]  \notag \\
&= e^{- n \ell ( \|  y  \| ) T \cdot \snr^{-1} } \mathbb{E}_{\Phi_m} [\prod_j    \frac{1}{ (1 + \ell ( \|  y  \| )  T  / \ell ( \| x_j - y   \| ) )^n } ]  \notag \\
&= e^{- n \ell ( \|  y  \| ) T \cdot \snr^{-1} } \mathbb{E}_{\Phi_m} \left[ e^{\sum_j \log   \frac{1}{ (1 + \ell ( \|  y  \| )  T  / \ell ( \| x_j - y   \| ) )^n } } \right]  \notag \\
&= e^{- n \ell ( \|  y  \| ) T \cdot \snr^{-1} }  e^{ - \lambda_m \int_{\mathbb{R}^2} 1 -  \frac{1}{ (1 + \ell ( \|  y  \| )  T  / \ell ( \| x - y   \| ) )^n } \dint x }, \notag  
\end{align}
where the last equality follows from the Laplace functional of the PPP $\Phi_m: \cc L_{\Phi_m} (f) = \exp ( - \lambda_m \int_{\mathbb{R}^2} (1 - e^{-f(x)}) \dint x  ) $ where $f: \mathbb{R}^2 \to \mathbb{R}^+$ \cite{baccelli2009stochastic}. Further, by the staionarity of the PPP $\Phi_m$ and changing Cartesian coordinates to Polar coordinates, 
\begin{align}
p_n (y ) = e^{-n \ell ( \|  y  \| ) T \cdot \snr^{-1} }  e^{ - 2\pi \lambda_m \int_{0}^\infty ( 1 -   (1 + T\ell ( \|  y  \| )  / \ell ( r ) )^{-n}  ) r \dint r }.  \label{eq:app:1}
\end{align}
To sum up,
\begin{align}
p (y  ) &= \mathbb{E}[ p (y | \Phi_m  ) ] = \sum_{n=1}^{\tau_m} (-1)^{n+1} \binom{\tau_m}{n} \mathbb{E} E[ p_n (y | \Phi_m  ) ]  \notag \\
&= \sum_{n=1}^{\tau_m} (-1)^{n+1} \binom{\tau_m}{n} p_n (y ). \notag   
\end{align}
Plugging the explicit expression (\ref{eq:app:1}) for $p_n (y   )$ completes the proof.

\subsection{Proof of Proposition \ref{pro:1}}
\label{proof:pro:1}

We first evaluate $p_n (y_1, y_2)$  defined as follows:
$$p_n (y_1, y_2) \triangleq  \mathbb{P}^o (   (\cap_{m=1}^n E(y_1)) \cap (\cap_{m=1}^n E(y_2) )  ) .
$$
As in the proof of Theorem \ref{thm:1}, the calculation consists of two steps: first evaluate $p_n (y_1, y_2 | \Phi_m  )$ conditioned on $\Phi_m$, and then de-condition on $\Phi_m$ to obtain $p_n (y_1, y_2)$. 

Following similar arguments in the proof of Theorem \ref{thm:1}, the conditional $p_n (y_1, y_2 | \Phi_m  )$ can be calculated as
\begin{align}
p_n (y_1, y_2 | \Phi_m  ) &= e^{ -n (\ell ( \|  y_1  \| ) + \ell ( \|  y_2  \| )) T \cdot \snr^{-1} } \cdot   \notag \\
& \quad \prod_{i=1}^2\prod_j    \frac{1}{ (1 + \ell ( \|  y_i  \| )  T  / \ell ( \| x_j - y_i   \| ) )^n }. \notag 
\end{align}
Now de-conditioning on $\Phi_m$ yields
\begin{align}
&p_n (y_1, y_2   ) = \mathbb{E}_{\Phi_m} [p_n (y_1, y_2 | \Phi_m  )  ] \notag \\
&= e^{ -n (\ell ( \|  y_1  \| ) + \ell ( \|  y_2  \| )) T \cdot \snr^{-1} } \cdot   \notag \\
\quad & e^{ - \lambda_m \int_{\mathbb{R}^2} 1 - \frac{1}{ (1 + \ell ( \|  y_1  \| )  T  / \ell ( \| x - y_1   \| ) )^n } \frac{1}{ (1 + \ell ( \|  y_2  \| )  T  / \ell ( \| x - y_2   \| ) )^n } \dint x } . \notag 
\end{align}
Using $p_n (y_1, y_2)$, the conditional probability $p_n (y_1 | y_2)$ can be calculated as follows:
\begin{align}
& p_n (y_1 | y_2) 
= \frac{p_n (y_1 , y_2)}{ p_n (  y_2) }  
= e^{- n \ell ( \|  y_1  \| ) T \cdot \snr^{-1} } \cdot   \notag \\
& \quad e^{ - \lambda_m \int_{\b R^2} \frac{1}{ (1 + \ell ( \|  y_2  \| )  T  / \ell ( \| x - y_2   \| ) )^n } \left(1 -  \frac{1}{ (1 + \ell ( \|  y_1  \| )  T  / \ell ( \| x - y_1   \| ) )^n } \right) \dint x } . \notag 
\end{align}
This completes the proof.

\subsection{Proof of Proposition \ref{pro:2}}
\label{proof:pro:2}
Define the typical coverage cell $\cc A_{o}$ of the multicast transmitter $x_o \in \Phi_m$ as
\begin{align}
\cc A_{o} = \left\{ y \in B(x_i, R): \exists  n \textrm{ s.t. }\sinr_{y,o} (n) \geq T \right\}. \notag
\end{align}
We next establish the relation between $\mathbb{E}^o[N]$ and the mean cell volume $\mathbb{E}^o [ | \cc A_o | ]$:  $\mathbb{E}^o[ N ] = \lambda_r \mathbb{E}^o [ | \cc A_o | ]$. To this end, 
\begin{align}
\mathbb{E}^o[ N ] &=\mathbb{E}^o [ \sum_{y \in \Phi_{m,o}} \mathbb{I} (\{ y \textrm{ is covered} \}) ] \notag \\
&=  \mathbb{E}^o \left[ \lambda_r \int_{\mathbb{R}^2} \mathbb{I}(\{ y \textrm{ is covered} \}) \dint y \right] \notag \\
&=  \lambda_r \int_{\mathbb{R}^2}  \mathbb{E}^o \left[\mathbb{I}(\{ y \textrm{ is covered} \}) \right]  \dint y \notag \\
&=  \lambda_r \int_{\mathbb{R}^2}  p(y) \dint y, \notag
\end{align}
where the second and third equalities follow from Campbell's theorem \cite{baccelli2009stochastic} and Fubini's theorem, respectively. Similarly, we have
\begin{align}
\mathbb{E}^o [ | \cc A_o | ] &= \mathbb{E}^o \left[ \int_{\mathbb{R}^2} \mathbb{I}( \{ y \in \cc C_0 \} ) \dint y \right] 
 =  \int_{B(0,R)}   p_0 (y)  \dint y. \notag
\end{align}
It follows that  $\mathbb{E}^o[ N ] = \lambda_r \mathbb{E}^o [ | \cc A_o | ]$.
The proof will be complete once we compute $\mathbb{E}^o [ | \cc A_0 | ]$. To this end, using Corollary \ref{cor:1}, we obtain that $\mathbb{E}^o [ | \cc A_0 | ]$ equals
\begin{align}
&2\pi  \int_0^{R}  \sum_{n=1}^{\tau_m} (-1)^{n+1} \binom{\tau_m}{n} e^{- n  T \cdot \snr^{-1} A r^{\alpha}  - \lambda_m K(\alpha, n) T^{\frac{2}{\alpha}} r^2  }  r \textrm{d} r \notag 
\end{align}
By Fubini's theorem we can exchange the above summation and integration. Then invoking the established relation $\mathbb{E}^o[ N ] = \lambda_r \mathbb{E}^o [ | \cc A_o | ]$ completes the proof.

\subsection{Proof of Proposition \ref{pro:8}}
\label{proof:pro:8}

By definition, the fraction of null receiver clusters equals the probability that the typical cluster  is in bad geometry. Denoting this event by $E_{bad}$, we consider the following two cases.

If $R<R_{th}$, then $E_{bad}$ is equivalent to the event there exists no point in the typical cluster $\Phi_{m,x_0}$. By Poisson assumption, we have $\Phi_{m,x_0}( B(o, R) ) \sim \textrm{Poisson} ( \lambda_r \pi R^2 ) $. Then,
$$
\mathbb{P}^o (E_{bad})  =  \mathbb{P}^o ( \Phi_{m,x_0}( B(o, R) ) = 0  ) = e^{-\lambda_r \pi R^2 }.
$$

If $R\geq R_{th}$, denoting by $A_n = \{ \Phi_{m,x_0}( B(o, R) ) = n \}$, 
\begin{align}
&\mathbb{P}^o (E_{bad})  = \sum_{n=0}^{\infty} \mathbb{P}^o ( A_n  ) \mathbb{P}^o (E_{bad} | A_n ) \notag \\
&=  \sum_{n=0}^{\infty}  \frac{ (\lambda_r \pi R^2)^n e^{-\lambda_r \pi R^2} }{n!} \left( \frac{\lambda_r \pi R^2- \lambda_r \pi R_{th}^2 }{\lambda_r \pi R^2} \right)^n \notag \\
&= e^{-\lambda_r \pi R_{th}^2 }, \notag 
\end{align}
where the second equality follows from that conditioning on $\Phi_{m,x_0}( B(o, R) ) = n$ these $n$ points are i.i.d. uniformly distributed in $B(o, R)$. To sum up, we have $\mathbb{P}^o (E_{bad}) = e^{ -\lambda_r \pi (\max(R_{th}, R))^2 }$.

\subsection{Proof of Proposition \ref{pro:9}}
\label{proof:pro:9}

When $W \equiv 0$ and $\lambda_m \to \infty$, by Corollary \ref{cor:2} 
$$
\xi = \frac{\pi  \tilde{K} (\alpha, \tau_m)\lambda_r}{T^{\frac{2}{\alpha}} \lambda_m} \cdot \frac{1}{\tau_m} \log(1+T).
$$
It follows that $\textrm{maximize}_{T > 0}  \xi$ is equivalent to
$
\textrm{maximize}_{T > 0}  T^{-\frac{2}{\alpha}} \log(1+T). \notag 
$
Let $f(x) = x^{-\frac{2}{\alpha}} \log(1+x)$. Direct calculation yields 
$$
\frac{\textrm{d} f }{\textrm{d} x} = \frac{1}{x^{\frac{2}{\alpha}+1}} \left( \frac{x}{1+x} - \frac{2}{\alpha} \log(1+x) \right). 
$$
Denote by $g(x)$ the term inside the above parentheses. Direct calculation yields 
$
\frac{\textrm{d} g }{\textrm{d} x} = \frac{1}{1+x} \left( \frac{1}{1+x} - \frac{2}{\alpha}  \right) .
$
It follows that $\frac{\textrm{d} g }{\textrm{d} x} > 0$ when $x \in (0,\frac{\alpha}{2} -1)$ and $\frac{\textrm{d} g }{\textrm{d} x} < 0$ when $x \in (\frac{\alpha}{2} -1, \infty)$. Correspondingly, $g(x)$ first monotonically increases from $0$ to $g(\frac{\alpha}{2} -1) = 1-\frac{2}{\alpha} (1+ \log(\frac{\alpha}{2}) )$ (which is positive when $\alpha >2$), and then monotonically decreases from $g(\frac{\alpha}{2} -1)$ to $-\infty$. Thus, there exists a unique point $x^\star > \frac{\alpha}{2} -1$ such that $g(x^\star = 0)$, and $f(x)$ monotonically increases when $x\in (0,x^\star)$ and then decreases when $x\in (x^\star,\infty)$. The last fact implies that $x^\star$ is the unique optimal point and completes the proof.

\subsection{Proof of Proposition \ref{pro:3}}
\label{proof:pro:3}

The proof is similar to that of Theorem \ref{thm:1} except the following arguments:
\begin{align}
p_n (y  ) &= \mathbb{P}^o (  \sinr_{y,0} (k) \geq T, \forall k \in \{1,...,n\}  )  \notag \\
&= \prod_{k=1}^{n} \mathbb{P}^o (  \sinr_{y,0} (k) \geq T    ) \notag 
\end{align}
where the last equality follows from the fact that an independent PPP $\Phi_m(k)$ is drawn at each time slot, and $\mathbb{P}^o (  \sinr_{y,0} (k) \geq T    )$ equals
\begin{align}
&\mathbb{P}^o \left(  \frac{  F_{y,x_0} (k) / \ell ( \|  y  \| ) }{ \snr^{-1}  +\sum_{x_j \in \Phi_m(k): j \neq 0} F_{y,x_j} (k) / \ell ( \| x_j -  y   \| ) } \geq T     \right ) \notag \\
&= e^{- \ell ( \|  y  \| ) T \cdot \snr^{-1} } \cdot \notag \\   &\quad \mathbb{E}^o_{\Phi_m(k),F} \left[ e^{ -  \ell ( \|  y  \| )  T \sum_{j \in \Phi_m (k)} F_{y,x_j} (k) / \ell ( \| x_j - y   \| ) }    \right ] \notag  \\
&=e^{- n\ell ( \|  y  \| ) T \cdot \snr^{-1} }\prod_{k=1}^{n}e^{ - \lambda_m \int_{\mathbb{R}^2} 1 -  \frac{1}{ (1 + \ell ( \|  y  \| )  T  / \ell ( \| x - y   \| ) ) } \dint x } \notag \\
&= e^{-  n T \cdot \snr^{-1} A \|  y  \|^{\alpha} }  e^{ - \lambda_m nK(\alpha, 1) T^{\frac{2}{\alpha}} \| y \|^2  }, \notag
\end{align}
Then $p_0 (y  )$ can be readily obtained by plugging $p_n (y  )$ into the equality $p_0 (y  )= \sum_{n=1}^{\tau_m} (-1)^{n+1} \binom{\tau_m}{n} p_n (y   ) $. Also, the mean number of covered receivers can be evaluated using the equality $\mathbb{E}^o[ N ] =   \lambda_r \int_{B(0,R)} p_0 (y) \dint y$, which has been established in the proof of Prop. \ref{cor:2}.

\subsection{Proof of Lemma \ref{lem:1}}
\label{proof:lem:1}

By definition, we have
\begin{align}
 &n K(\alpha, 1) - K(\alpha, n) = n \frac{2\pi }{\alpha}\int_0^{\infty} t^{-\frac{2}{\alpha} - 1} \left( 1 - \frac{1}{1+t} \right) \dint t    \notag \\
 &\quad  -\frac{2\pi }{\alpha}\int_0^{\infty} t^{-\frac{2}{\alpha} - 1} \left( 1 - \frac{1}{(1+t)^n} \right) \dint t \notag \\
&=  \frac{2\pi }{\alpha}\int_0^{\infty} t^{-\frac{2}{\alpha} - 1} \left( n- 1 - \frac{n}{1+t}  + \frac{1}{(1+t)^n} \right) \dint t. \notag 
\end{align}
Denote by $f(t) = n- 1 - \frac{n}{1+t}  + \frac{1}{(1+t)^n}, t \geq 0$. Note $f'(t) = \frac{(1+t)^{n-1} - 1}{(1+t)^{n+1}} \geq 0, \forall t \geq 0$. It follows that $f(t)$ is monotonically increasing on $[0,\infty]$ and $f(t) \geq f(0) = 0$. Thus $n K(\alpha, 1) - K(\alpha, n) = \frac{2\pi }{\alpha}\int_0^{\infty} t^{-\frac{2}{\alpha} - 1} f(t) \dint t \geq 0$ as the integrand is non-negative.

%

\subsection{Proof of Proposition \ref{pro:6}}
\label{proof:pro:6}

Under the Palm measure (with respect to $\Phi_b$),
\begin{align}
&\mathbb{E}^o [  \sum_{x\in \Phi_m \cap \mathcal{C}_o (\Phi_b)}  \sum_{y \in \Phi_{m,x}}  \mathbb{I}( \{y \textrm{ is covered}\})  ]   \notag \\
&= \mathbb{E}^o [  \sum_{x\in \Phi_m \cap \mathcal{C}_o (\Phi_b)} \mathbb{E}[ \sum_{y \in \Phi_{m,x}}  \mathbb{I}( \{y \textrm{ is covered}\})  ] ] .
\label{eq:31}
\end{align}
For $x \in \mathcal{C}_o (\Phi_b) \cap \Phi_m $, using Corollary \ref{cor:3} and Prop. \ref{pro:5}, we have 
\begin{align}
\mathbb{E}[ \sum_{y \in \Phi_{m,x}}  \mathbb{I}( \{y \textrm{ is covered}\})  ] = h( \| x \| ; \tau_m, g_o (\|x\|) ). \notag 
\label{eq:30}
\end{align}
So (\ref{eq:31}) characterizes the mean additive characteristic with functional $h( \cdot )$ \cite{baccelli2009stochastic} and equals
\begin{align}
&\mathbb{E}^o [  \sum_{x\in \Phi_m \cap \mathcal{C}_o (\Phi_b)} h( \| x \| ; \tau_m, g_o (\| x \|) )  ]  \notag \\
&=\mathbb{E}^o \left[ \lambda_m \int_{\mathbb{R}^2} \mathbb{I} ( x \in \mathcal{C}_o (\Phi_b) ) \cdot h( \| x \| ; \tau_m, g_o (\| x \|) ) \dint x  \right] \notag \\
&=  \lambda_m \int_{\mathbb{R}^2} \mathbb{E}^o \left[ \mathbb{I} ( x \in \mathcal{C}_o (\Phi_b) ) \right] \cdot h( \| x \| ; \tau_m, g_o (\| x \|) ) \dint x  \notag \\
&=  \lambda_m \int_{\mathbb{R}^2} \mathbb{P}^o ( x \in \mathcal{C}_o (\Phi_b) ) \cdot h( \| x \| ; \tau_m, g_o (\| x \|) ) \dint x  \notag \\
&=  \lambda_m \int_{\mathbb{R}^2} \mathbb{P}^o (  \Phi_b ( B^0 (x, \|x\|) ) = 0 ) \cdot h( \| x \| ; \tau_m, g_o (\| x \|) ) \dint x  \notag \\
&=  \lambda_m \int_{\mathbb{R}^2}  e^{-\lambda_b \pi \|x\|^2} \cdot h( \| x \| ; \tau_m, g_o (\| x \|) ) \dint x  \notag \\
&= 2 \pi \lambda_m \int_{\mathbb{R}^+}  e^{-\lambda_b \pi r^2} \cdot h( r ; \tau_m, g_o (r) ) r \dint r  \notag \\
&= \frac{ \lambda_m }{ \lambda_b } \mathbb{E}_D \left[  h( D ; \tau_m, g_o (D) ) \right], \notag 
\end{align}
where we use Campbell theorem in the first equality and Fubini theorem in the second equality; the fourth equality follows since  $x \in \mathcal{C}_o (\Phi_b)$ if and only if $o$ is the nearest BS in $\Phi_b$, i.e., $\Phi_b ( B^0 (x, \|x\|) ) = 0$; the fifth equality follows from the fact $\Phi_b ( B^0 (x, \|x\|) ) \sim \textrm{Poisson}(\lambda_m \pi \|x\|^2)$; and we convert from Cartesian to polar coordinates in the penultimate equality. Using similar arguments, we can derive the other two expectation terms; we omit them for brevity.

\subsection{Proof of Proposition \ref{pro:7}}
\label{proof:pro:7}

Suppose $\mathcal{O} = (\tau_m^{\star}, \{ g^{\dagger}_z ( \| x_i - z \|  ) \} )$ is an optimal solution but does not satisfy $g^{\dagger}_z ( \| x_1 - z \|  ) \geq .... \geq  g^{\dagger}_z ( \| x_{M_z} - z \|  )$. Then $\mathcal{O}$ has at least one pair $(i,j)$ such that $1\leq i < j \leq M_z$ and
$
0 = g^{\dagger}_z ( \| x_i - z \|  ) <  g^{\dagger}_z ( \| x_j - z \|  ) = 1.
$
We will decrease the number of such pairs in $\mathcal{O}$ by swapping the values of the binary decision variables: $g^{\dagger}_z ( \| x_i - z \|  ) = 1$ and $g^{\dagger}_z ( \| x_j - z \|  ) = 0$. We denote by $\tilde{\mathcal{O}} = (\tau_m^{\star}, \{ \tilde{g}_z ( \| x_i - z \|  ) \} )$ the solution after the swapping. First, we claim $\tilde{\mathcal{O}}$ is feasible; indeed, $\tilde{g}_z (r_k) \in \{0,1\}$, $\sum_{k=1}^{M_z}  \tilde{g}_z (r_k) = \sum_{k=1}^{M_z}  g^{\dagger}_z ( r_k ) \leq B$, and using the fact that $q(r)$ is strictly decreasing with $r$,
\begin{align}
&\frac{1}{ M_z } \sum_{k=1}^{M_z}  h\left( r_k ; \tau^\star_m, \tilde{g}_z (r_k) \right) - \frac{1}{ M_z } \sum_{k=1}^{M_z}  h\left(r_k ; \tau^\star_m, g^{\dagger}_z (g_k) \right)  \notag \\
&= \frac{1}{ M_z } (q ( r_i ) - q(r_j) ) ( \bar{N}_{\max} -  \bar{N} (\tau_m)) > 0, \notag 
\end{align}
which shows that $\tilde{\mathcal{O}}$ meets all the constraints.
Further, $\tilde{\mathcal{O}}$ gives the optimal objective value $\tau_m^{\star}$; and thus $\tilde{\mathcal{O}}$ is also an optimal solution. Repeating iteratively the above exchange arguments, we can construct an optimal solution  such that $g^{\star}_z ( r_1  ) \geq .... \geq  g^{\star}_z ( r_z  )$. This completes the proof.

\bibliographystyle{IEEEtran}
\bibliography{IEEEabrv,Reference}

\end{document}